\documentclass[a4paper,11pt]{article}

\usepackage{fullpage}
\usepackage{times}
\usepackage{soul}
\usepackage{url}
\usepackage[utf8]{inputenc}
\usepackage[small]{caption}
\usepackage{graphicx}
\usepackage{xcolor}
\usepackage{amsmath}
\usepackage{booktabs}
\usepackage[linesnumbered,ruled,vlined]{algorithm2e}
\usepackage{enumitem}
\usepackage{dsfont}
\usepackage{enumerate}
\usepackage{multirow}

\usepackage{xcolor}
\definecolor{darkgreen}{rgb}{0,0.5,0}
\usepackage{hyperref}
\hypersetup{
    unicode=false,          
    colorlinks=true,        
    linkcolor=red,          
    citecolor=darkgreen,    
    filecolor=magenta,      
    urlcolor=blue           
}
\usepackage{cleveref}

\usepackage{amsthm}
\usepackage{amssymb}
\usepackage{algpseudocode}

\usepackage{bbm}

\newtheorem{theorem}{Theorem}[section]

\newtheorem{proposition}[theorem]{Proposition}
\newtheorem{lemma}[theorem]{Lemma}

\theoremstyle{definition}
\newtheorem{definition}[theorem]{Definition}
\newtheorem{example}[theorem]{Example}

\usepackage{todonotes}

\usepackage{natbib}

\allowdisplaybreaks

\usepackage{authblk}

\title{\bf Envy-Free House Allocation with Minimum Subsidy}

\author[1]{Davin Choo}
\author[1]{Yan Hao Ling}
\author[1]{Warut Suksompong}
\author[2]{Nicholas Teh}
\author[1]{Jian Zhang}

\affil[1]{National University of Singapore, Singapore}
\affil[2]{University of Oxford, UK}

\date{\vspace{-10mm}}

\begin{document}

\maketitle

\begin{abstract}
House allocation refers to the problem where $m$ houses are to be allocated to $n$ agents so that each agent receives one house. 
Since an envy-free house allocation does not always exist, we consider finding such an allocation in the presence of subsidy. 
We show that computing an envy-free allocation with minimum subsidy is NP-hard in general, but can be done efficiently if $m$ differs from~$n$ by an additive constant or if the agents have identical utilities.
\end{abstract}

\section{Introduction}

House allocation is a classic economic problem where $m$ houses are to be allocated to $n\le m$ agents with possibly differing preferences so that each agent receives exactly one house \citep{abdulkadirouglu1998random,atila1999house,hylland1979positions,zhou1990house}. 
Applications of this problem include the assignment of students to dormitory rooms, company workers to offices, and clients to web servers.

An important consideration in house allocation is fairness, which is often captured by the fundamental notion of \emph{envy-freeness} \citep{foley1967resource,varian1974equity}.
An allocation is said to be envy-free if no agent prefers a different agent's house to her own.
Unfortunately, an envy-free house allocation does not always exist, for instance, when there are two houses and two agents who both have value $200$ for the first house and $100$ for the second house.
\citet{gan2019envy} presented a polynomial-time algorithm that decides whether an envy-free house allocation exists and, if so, computes one such allocation.
Other works on house allocation have examined measures such as the number of envious agents or the aggregate envy among agents.
In particular, \citet{kamiyama2021complexity} showed that the problem of maximizing the number of envy-free agents is hard even to approximate, while \citet{madathil2023complexity} studied the complexity of several envy-related measures.
\citet{beynier2019local} and \citet{hosseini2023graphical} assumed that agents can only envy other agents with whom they are acquainted according to a given social graph; they investigated the case $m = n$ and derived complexity results for a range of graphs.

In order to circumvent the possible non-existence of envy-free allocations, we consider utilizing money in the form of \emph{subsidy}.
Specifically, we may give each agent some amount of subsidy in addition to the house.
For instance, in the example from the previous paragraph, if we assign the first house to one agent and the second house along with a subsidy of $100$ to the other agent, the resulting allocation is envy-free: each agent has value $200$ for her own bundle as well as for the other agent's bundle.
As we shall later elaborate, it follows from known results that with a sufficient amount of subsidy, attaining envy-freeness among the agents is always possible.
Therefore, we focus on the natural goal of achieving envy-freeness by using the \emph{minimum} subsidy.
The aforementioned result of \citet{gan2019envy} implies that we can decide in polynomial time whether there exists an envy-free house allocation with zero subsidy.
Can we also efficiently determine the smallest amount of subsidy when this amount is nonzero?

In this note, we provide a negative answer to the question above: perhaps surprisingly, finding the minimum subsidy required to achieve envy-freeness is NP-hard in general. 
This means that there does not exist a polynomial-time algorithm for this problem unless P = NP.
To complement this result, we show that the problem becomes polynomial-time solvable in two special cases: when $m$ differs from $n$ by an additive constant (including when $m = n$) and when the agents have identical utilities.

Before proceeding further, we briefly discuss two other related lines of work.
\begin{itemize}
\item The use of subsidy has received interest in the \emph{fair allocation of indivisible goods} \citep{halpern2019fair,brustle2020one,caragiannis2021subsidy,barman2022achieving,goko2024subsidy}.
In that setting, each agent may receive any number of goods, and it is usually required that all goods be allocated.
Observe that with these assumptions, deciding whether there exists an envy-free allocation with zero subsidy is already NP-hard---the hardness holds even for two agents with identical utilities, by a trivial reduction from the \textsc{Partition} problem.
By contrast, this task can be solved efficiently in house allocation \citep{gan2019envy}.

\item A number of authors have studied settings where $m = n$ goods along with a monetary asset or liability are divided among $n$ agents \citep{maskin1987division,klijn2000algorithm,sun2003sphouse,svensson1983rental}; in the case of a liability, the problem is also known as \emph{rent division} \citep{su1999rentalharmony,gal2017fairest}.
In these works, the monetary amount is fixed in advance, unlike in our work where we aim to optimize the amount subject to achieving envy-freeness.
\end{itemize}

\section{Preliminaries}

Let $[k] := \{1,2,\dots,k\}$ for any positive integer~$k$.
In the house allocation problem, we are given a set $N = [n]$ of $n$ agents and a set $H = \{h_1,\dots,h_m\}$ of $m$ houses; each agent $i \in N$ has a nonnegative utility $u_i(h_j)$ for each house $h_j\in H$.
A house allocation \emph{problem instance} consists of the set of agents, the set of houses, and the agents' utilities for the houses.
An \emph{allocation} $\mathbf{a} = (a_1,\dots,a_n)$ is a list of $n$ distinct houses, where house $a_i\in H$ is allocated to agent~$i$. 
The subsidy given to the agents is represented by a nonnegative \emph{subsidy vector} $\mathbf{s} = (s_1,\dots,s_n)$.
We refer to an allocation $\mathbf{a}$ with a corresponding subsidy vector $\mathbf{s}$ as an \emph{outcome} $(\mathbf{a},\mathbf{s})$. 
Then, we define an envy-free outcome as follows.

\begin{definition}
    An outcome $(\mathbf{a},\mathbf{s})$ is \emph{envy-free} 
    if $u_i(a_i) + s_i \geq u_i(a_{i'}) + s_{i'}$ for every pair of agents $i, i' \in N$.
\end{definition}

We say that a subsidy vector $\mathbf{s}$ is \emph{envy-eliminating} for an allocation $\mathbf{a}$ if the outcome $(\mathbf{a}, \mathbf{s})$ is envy-free.
Not every allocation admits an envy-eliminating subsidy vector, as illustrated by the following example.

\begin{example}
\label{ex:envy-freeable}
    Suppose there are $m = 2$ houses and $n = 2$ agents whose utilities are given by $u_1(h_1)= u_1(h_2)=u_2(h_1)=200$ and $u_2(h_2)=100$.
    We claim that there does not exist an envy-eliminating subsidy vector for the allocation $\mathbf{a} = (h_1,h_2)$.
    To see this, let $\mathbf{s} = (s_1,s_2)$ be any subsidy vector. In order for the outcome $(\mathbf{a},\mathbf{s})$ to be envy-free, the following must hold:
    \begin{align*}
        200 + s_1 = u_1(h_1) + s_1 & \geq u_1(h_2) + s_2 = 200 + s_2, \text{ and}\\
        100 + s_2 = u_2(h_2) + s_2 & \geq u_2(h_1) + s_1 = 200 + s_1.
    \end{align*}
    This implies that $s_1 \geq s_2 \geq 100+ s_1$, which is impossible.
\end{example}

To capture the phenomenon in Example~\ref{ex:envy-freeable}, we use the concept of envy-freeability, introduced by \citet{halpern2019fair} in the context of fairly allocating indivisible goods.

\begin{definition}
    An allocation $\mathbf{a}$ is \emph{envy-freeable} if there exists a subsidy vector $\mathbf{s}$ such that the outcome $(\mathbf{a}, \mathbf{s})$ is envy-free.
\end{definition}

While the allocation in Example~\ref{ex:envy-freeable} is not envy-freeable, a result by \citet[Theorem~1]{halpern2019fair} implies that every house allocation problem instance (with arbitrary $n$ and $m$) admits an envy-freeable allocation. 
Among all envy-freeable allocations, we are interested in finding one that requires the least amount of subsidy.
We say that $(\mathbf{a},\mathbf{s})$ is a \emph{minimum-subsidy envy-free outcome} if it is an envy-free outcome and for every allocation $\mathbf{a}'$ (possibly $\mathbf{a}' = \mathbf{a}$) with a corresponding envy-eliminating subsidy vector $\mathbf{s}'$, it holds that $\sum_{i \in N} s_i \leq \sum_{i \in N} s'_i$.

\section{NP-Hardness}

As mentioned earlier, \citet{gan2019envy} proved that determining whether there exists an envy-free house allocation with zero subsidy can be done in polynomial time.
Our main result shows that their algorithm \emph{cannot} be extended to arbitrary subsidy thresholds, provided that P $\ne$ NP.

\begin{theorem} \label{thm:hardness}
    The problem of computing a minimum-subsidy envy-free  outcome is \emph{NP}-hard. 
\end{theorem}
\begin{proof}
    Given a house allocation problem instance and a threshold~$\gamma > 0$, consider the problem of deciding whether there exists an envy-free outcome $(\mathbf{a},\mathbf{s})$ such that $\sum_{i \in N} s_i \leq \gamma$.
    We will show that this decision problem is NP-hard; this immediately implies the NP-hardness of computing a minimum-subsidy envy-free outcome.
    
    We reduce from the NP-hard problem {\sc Vertex Cover}~\citep{karp1972complexity}.
    Given an undirected graph $G = (V,E)$, a subset $V'\subseteq V$ of vertices is said to be a \emph{vertex cover} of $G$ if every edge in $E$ contains at least one vertex in $V'$.
    An instance of \textsc{Vertex Cover} consists of an undirected graph $G = (V,E)$ and a positive integer $k$; it is a Yes-instance if $G$ admits a vertex cover of size at most $k$, and a No-instance otherwise.

    Consider an arbitrary instance of {\sc Vertex Cover}.
    Since any subset of $|V|-1$ vertices is trivially a vertex cover, we may assume that $k < |V| - 1$.
    We construct a house allocation problem instance with $n = |V|^4 + |V|^3 + |E|$ agents and $m = |V|^4 + |V|^3 + |V|^2$ houses.
    Note that since $|E| \leq \binom{|V|}{2} < |V|^2$, we have $n < m$.
    First, we specify the two classes of houses.
    \begin{itemize}
        \item \textbf{Special Houses}: Let there be $|V|^4$ special houses.
        \item \textbf{Vertex Houses}: For each vertex $v \in V$, let there be $|V|$ vertex houses of type $v_\text{good}$ and $|V|^2$ vertex houses of type $v_\text{bad}$.
        (Note that $v_\text{good}$ is not a single type but rather corresponds to $|V|$ different types, one for each $v\in V$; the same holds for $v_\text{bad}$.)
        Hence, the total number of vertex houses is $|V|^3 + |V|^2$.
    \end{itemize}    
    Next, we specify the three classes of agents along with their utilities.
    \begin{itemize}
        \item \textbf{Special Agents}: Let there be $|V|^4$ special agents.
        Each special agent has utility~$1$ for each special house and utility~$0$ for every other house.
        \item \textbf{Edge Agents}: For each edge $e  \in E$, let there be one edge agent corresponding to $e$.
        Hence, the total number of edge agents is $|E|$. 
        If $e = \{x,y\}$, the edge agent corresponding to $e$ has utility~$1$ for each special house as well as for each vertex house of type $x_\text{good}$ or $y_\text{good}$, and utility~$0$ for every other house.
        \item \textbf{Vertex Agents}: For each vertex $w \in V$, let there be $|V|^2$ vertex agents of type $w$. 
        Hence, the total number of vertex agents is $|V|^3$.
        Each agent of type $w$ has utility $1+|V|^{-3}$ for each vertex house of type $w_\text{good}$, utility~$1$ for each vertex house of type $w_\text{bad}$, and utility~$0$ for every other house.
    \end{itemize}
The constructed house allocation instance is summarized in \Cref{tab:vertex-cover-reduction}; clearly, the construction takes polynomial time.

\begin{table}[htb]
\centering
\begin{tabular}{ccccc}
\toprule
&& \multicolumn{3}{c}{Houses}\\
\cmidrule{3-5}
&& \begin{tabular}{cc}
Special \\($|V|^4$) 
\end{tabular} & 
\begin{tabular}{cc}
Vertex
$v_\text{good}$ \\($|V|$ for each $v$) 
\end{tabular} & 
\begin{tabular}{cc}
Vertex
$v_\text{bad}$ \\($|V|^2$ for each $v$) 
\end{tabular}\\
\midrule
\multirow{6}{*}{\rotatebox[origin=c]{90}{Agents}}
& \multicolumn{1}{|c}{Special ($|V|^4$)} & $1$ & $0$ & $0$\\
& \multicolumn{1}{|c}{
\begin{tabular}{cc}
Vertex
$w$ \\($|V|^2$ for each $w\in V$) 
\end{tabular}} & $0$ &
$
\begin{cases}
1 + |V|^{-3}  &\text{if } v = w\\
0  &\text{otherwise}
\end{cases}
$
&
$\begin{cases}
1  &\text{if } v = w\\
0  &\text{otherwise}
\end{cases}$
\\
& \multicolumn{1}{|c}{
\begin{tabular}{cc}
Edge
$e = \{x,y\}$ \\($1$ for each $e\in E$) 
\end{tabular}} & $1$ &
$
\begin{cases}
1 &\text{if } v\in\{x,y\}\\
0 &\text{otherwise}
\end{cases}
$
& $0$
\\
\bottomrule
\end{tabular}
\caption{House allocation problem instance constructed in the reduction of \Cref{thm:hardness}.}
\label{tab:vertex-cover-reduction}
\end{table}

We will argue that $G$ admits a vertex cover of size at most $k$ if and only if there exists an envy-free outcome in the corresponding house allocation problem with total subsidy at most $k/|V|$; this is sufficient to complete the proof.

\vspace{2mm}

($\Rightarrow$) Let $C\subseteq V$ be a vertex cover of $G$ such that $|C| \le k$. 
We will construct an envy-free outcome $(\mathbf{a},\mathbf{s})$ such that $\sum_{i \in N} s_i \leq k/|V|$.

First, choose an allocation $\mathbf{a}$ that satisfies the following properties:
\begin{itemize}
    \item Assign each special agent to a special house.
    \item Since $C$ is a vertex cover, for each edge agent corresponding to an edge $\{x,y\}$, at least one of $x$ and~$y$ belongs to $C$.
    Assign the edge agent to a vertex house of type $x_\text{good}$ or $y_\text{good}$ according to which of $x$ and $y$ belongs to $C$. 
    (If both $x$ and $y$ belong to $C$, choose arbitrarily.)
    Note that since there are $|V|$ vertex houses of type $v_\text{good}$ for each $v\in V$, each edge agent can indeed be assigned to a vertex house.
    \item Assign each vertex agent of type $v$ to a vertex house of type $v_\text{bad}$.
\end{itemize}
It is easy to see that such an allocation exists.
Next, we construct a subsidy vector $\mathbf{s}$ as follows:
\begin{equation*}
    s_i := 
    \begin{cases}
        |V|^{-3} & \text{if $i$ is a vertex agent of type $v$ for some $v \in C$}; \\
        0 & \text{otherwise}.
    \end{cases}
\end{equation*}
Since $|C| \leq k$ and there are $|V|^2$ vertex agents of each type $v \in C$, the total subsidy is 
\begin{equation*}
    \sum_{i \in N} s_i = \frac{|V|^2  \cdot |C|}{|V|^{3}} = \frac{|C|}{|V|} \leq \frac{k}{|V|}.
\end{equation*}

We now show that the outcome $(\mathbf{a},\mathbf{s})$ is envy-free.
Note that every agent has a utility of $1$ for her house.
Since special agents, edge agents, and vertex agents of each type $v \notin C$ have a maximum utility of~$1$ for every allocated house, they are envy-free among themselves.
In addition, these agents have utility~$0$ for every vertex house of type $v_\text{bad}$ for each $v \in C$, so they do not envy any vertex agent of type $v \in C$ (who is given a subsidy of $|V|^{-3} < 1$).
Finally, every vertex agent of type $v\in C$ has a utility of $1 + |V|^{-3}$ after subsidy, so these agents also do not envy any other agent.

\vspace{2mm}

($\Leftarrow$)
Let $(\mathbf{a},\mathbf{s})$ be an envy-free outcome with total subsidy at most $k/|V|$, that is,
\begin{equation} \label{eqn:nphard_contradiction}
    \sum_{i \in N} s_i \leq \frac{k}{|V|}.
\end{equation}
We will construct a vertex cover of $G$ of size at most $k$.
Let the set $T \subseteq V$ be defined by
\begin{equation*}
    T := \{ v \in V \mid \text{ there exists an edge agent receiving a house of type $v_\text{good}$ in $\mathbf{a}$}\}.
\end{equation*}
We claim that $T$ is a vertex cover of $G$ of size at most $k$.

First, we show that $T$ is a vertex cover of $G$.
Since $k < |V|-1$, we have that $\frac{k}{|V|} < \frac{|V|-1}{|V|} < 1$.
This means that $\sum_{i \in N} s_i \leq \frac{k}{|V|} < 1$, so no agent receives a subsidy of $1$ or more.

Since $n = |V|^4 + |V|^3 + |E| > |V|^3 + |V|^2 = m - |V|^4$, by the pigeonhole principle, some special house must be allocated.
If some special agent is not assigned to a special house, then a subsidy of at least~$1$ is required to make this agent envy-free towards an agent who is assigned to a special house, contradicting the assertion that each agent receives a subsidy of strictly less than~$1$.
Hence, every special agent must be assigned to a special house.

Now, since there are $|V|^4$ special houses and $|V|^4$ special agents, all special houses are allocated to special agents.
Hence, for each edge agent corresponding to an edge $\{x,y\} \in E$, in order not to require a subsidy of at least~$1$, the edge agent must be assigned to a vertex house of type $x_\text{good}$ or $y_\text{good}$.
By definition of $T$, this means that $T \cap \{x,y\} \neq \emptyset$ for every edge $\{x,y\} \in E$.
It follows that $T$ is a vertex cover of $G$.

It remains to show that $|T|\le k$. 
Consider an arbitrary vertex $v \in T$.
Since there exists an edge agent receiving a house of type $v_\text{good}$ (by definition of $T$), in order to avoid giving some vertex agent of type $v$ a subsidy of at least $1$, all vertex agents of type~$v$ must be assigned to a house of type $v_\text{good}$ or $v_\text{bad}$.
Since there are $|V|^2$ vertex agents of type $v$ but only $|V|$ vertex houses of type $v_\text{good}$, at least $|V|^2 - |V|$ of these agents require a subsidy of at least $|V|^{-3}$ each in order to be envy-free.
Thus, the total subsidy required is at least $|T| \cdot (|V|^2 - |V|) \cdot |V|^{-3}$. 
If $|T| \ge k+1$, we have
\begin{align*}
    \sum_{i \in N} s_i &\geq \frac{|T| \cdot (|V|^2 - |V|)}{|V|^3} 
    \geq \frac{(k+1) \cdot (|V|^2 - |V|)}{|V|^3} = \frac{1}{|V|} \cdot \left( k + 1 - \frac{k + 1}{|V|} \right) > \frac{k}{|V|},
\end{align*}
where the last inequality holds because $1 - \frac{k + 1}{|V|} > 0$, which follows from the assumption that $k < |V|-1$.
However, this is in conflict with (\ref{eqn:nphard_contradiction}).
It follows that $|T| \leq k$, as desired.
\end{proof}

We remark that the NP-hardness still holds even if we require the agents' utilities to be ``normalized'' in the sense that $\sum_{h\in H}u_i(h)$ is the same for all $i\in N$.
This is because we may modify each agent's utilities $u_i(h)$ to $\widehat{u}_i(h) := u_i(h) + c_i$ for some value $c_i \ge 0$; doing so does not affect whether a certain outcome is envy-free.

\section{Tractable Cases}

Theorem \ref{thm:hardness} reveals the general difficulty of finding a minimum-subsidy envy-free outcome in house allocation.
In this section, we show that the problem becomes tractable in two special cases.

\subsection{Similar Number of Agents and Houses}

We first consider the case where the number of agents and that of houses differ by an additive constant, and show the following result.

\begin{theorem}
\label{thm:similar-m-n}
When $m = n+c$ for some constant $c\ge 0$, a minimum-subsidy envy-free outcome can be computed in polynomial time.
\end{theorem}

Observe that when $m = n+c$, the number of possible subsets of houses to be allocated to the $n$~agents is $\binom{n+c}{n} \in O(n^c)$, which is polynomial in the input size.
Given this, we may find a minimum-subsidy envy-free outcome for each subset of $n$ houses, and output an outcome with the smallest subsidy across all such subsets.
Hence, to establish \Cref{thm:similar-m-n}, it suffices to prove the following proposition.

\begin{proposition}
\label{prop:equal-m-n}
When $m = n$, a minimum-subsidy envy-free outcome can be computed in polynomial time.
\end{proposition}

For a house allocation problem instance with $m = n$, consider a weighted bipartite graph $G$ with $n$ vertices on each side.
The vertices on one side correspond to the $n$ agents, the vertices on the other side correspond to the $n$ houses, and the weight of an edge is the utility that the corresponding agent has for the corresponding house.
Hence, a perfect matching of $G$ corresponds to a house allocation.
Results by \citet[Theorems~1 and 2]{halpern2019fair}, when translated to our house allocation setting, yield the following lemma.

\begin{lemma}[\citep{halpern2019fair}]
\label{lem:halpern-shah}
A perfect matching of $G$ is a maximum-weight perfect matching if and only if the corresponding allocation is envy-freeable.
Moreover, for any maximum-weight perfect matching of $G$ with a corresponding envy-freeable allocation~$\mathbf{a}$, there exists a polynomial-time algorithm that computes an envy-eliminating subsidy vector $\mathbf{s}$ for $\mathbf{a}$ such that for any other envy-eliminating subsidy vector $\mathbf{s}'$ for~$\mathbf{a}$, it holds that $\sum_{i\in N}s_i \le \sum_{i\in N}s'_i$.
\end{lemma}

Now, a graph $G$ may admit multiple maximum-weight perfect matchings, and \emph{a priori}, it is conceivable that one matching leads to a minimum-subsidy envy-free outcome while another does not.
Interestingly, we will show that this in fact cannot happen.
To do so, we leverage a result of \citet{barman2022achieving}.
Recall that a permutation $\sigma$ of $[n]$ is a bijection $\sigma : [n] \rightarrow [n]$.
For a permutation $\sigma$ of $[n]$ and an allocation $\mathbf{a}$, let $\mathbf{a}_\sigma := (a_{\sigma(1)},\dots,a_{\sigma(n)})$ be the allocation where each agent $i \in N$ receives the house $a_{\sigma(i)}$. 
Similarly, for a permutation $\sigma$ of $[n]$ and a subsidy vector $\mathbf{s}$, let $\mathbf{s}_\sigma := (s_{\sigma(1)},\dots,s_{\sigma(n)})$ be the subsidy vector where each agent $i \in N$ receives the subsidy $s_{\sigma(i)}$.

\begin{lemma}[Lemma~3 of \citep{barman2022achieving}]
\label{lem:barman}
Let $(\mathbf{a}, \mathbf{s})$ be an envy-free outcome, and let $\sigma$ be a permutation of $[n]$ such that the allocation $\mathbf{a}_\sigma$ is envy-freeable.
        Then, $(\mathbf{a}_\sigma, \mathbf{s}_\sigma)$ is an envy-free outcome as well.
\end{lemma}

\begin{lemma} \label{lem:m=n_permutation2}
        Let $(\mathbf{a},\mathbf{s})$ be a minimum-subsidy envy-free outcome, and let $\sigma$ be a permutation of $[n]$ such that the allocation $\mathbf{a}_\sigma$ is envy-freeable.
        Then, $(\mathbf{a}_\sigma, \mathbf{s}_\sigma)$ is a minimum-subsidy envy-free outcome as well.
\end{lemma}
\begin{proof}
        From Lemma~\ref{lem:barman}, we have that $(\mathbf{a}_\sigma,\mathbf{s}_\sigma)$ is an envy-free outcome.
        Since $\mathbf{s}$ and $\mathbf{s}_\sigma$ are permutations of each other, they give rise to the same total subsidy.
        Since $(\mathbf{a},\mathbf{s})$ is a minimum-subsidy envy-free outcome, we conclude that $(\mathbf{a}_\sigma,\mathbf{s}_\sigma)$ is a minimum-subsidy envy-free outcome as well.
\end{proof}

With all the ingredients in place, we are ready to prove Proposition~\ref{prop:equal-m-n}.

\begin{proof}[Proof of Proposition~\ref{prop:equal-m-n}]
Given a house allocation problem instance with $m = n$, construct a weighted bipartite graph $G$ as described earlier.   
Find a maximum-weight perfect matching of $G$; it is well-known that this can be done in polynomial time.
Let $\mathbf{a}$ be the corresponding allocation.
By Lemma~\ref{lem:halpern-shah}, $\mathbf{a}$ is envy-freeable, and there exists a polynomial-time algorithm that computes an envy-eliminating subsidy vector $\mathbf{s}$ for $\mathbf{a}$ that minimizes the total subsidy among all envy-eliminating subsidy vectors for $\mathbf{a}$.

We claim that $(\mathbf{a}, \mathbf{s})$ is a minimum-subsidy envy-free outcome.
It is clear from our construction that $(\mathbf{a}, \mathbf{s})$ is an envy-free outcome.
Consider a minimum-subsidy envy-free outcome $(\mathbf{a}', \mathbf{s}')$, where possibly $\mathbf{a}' = \mathbf{a}$ and/or $\mathbf{s}' = \mathbf{s}$.
Let $\sigma$ be the permutation such that $\mathbf{a}'_\sigma = \mathbf{a}$.
Since $\mathbf{a}'_\sigma$ is envy-freeable, Lemma~\ref{lem:m=n_permutation2} implies that $(\mathbf{a}'_\sigma, \mathbf{s}'_\sigma) = (\mathbf{a}, \mathbf{s}'_\sigma)$ is also a minimum-subsidy envy-free outcome.
Hence, by our construction of $\mathbf{s}$, the outcome $(\mathbf{a},\mathbf{s})$ must be a minimum-subsidy envy-free outcome as well.
\end{proof}

\subsection{Identical Utilities}

Next, we consider the case where all agents have the same utilities for houses, that is, $u_i(h) = u_{i'}(h)$ for all $i,i'\in N$ and $h\in H$.
We show that a minimum-subsidy envy-free outcome can also be computed efficiently in this case.
Let us denote the common utility function by $u$.

\begin{theorem}
    When all agents have identical utilities, a minimum-subsidy envy-free outcome can be computed in polynomial time.
\end{theorem}

\begin{proof}
Assume without loss of generality that the houses are ordered in non-decreasing order of utility: $u(h_1)\le\dots\le u(h_m)$.
Consider any set of $n$ houses $\{h_{j_1},\dots,h_{j_n}\}$, where $j_1\le\dots\le j_n$.
It is clear that the minimum subsidy required to achieve envy-freeness when allocating these $n$ houses is 
\begin{equation*}
(u(h_{j_n}) - u(h_{j_1})) + (u(h_{j_n}) - u(h_{j_2})) + \dots + (u(h_{j_n}) - u(h_{j_{n-1}}))
= n\cdot u(h_{j_n}) - \sum_{k = 1}^n u(h_{j_k}).
\end{equation*}
Hence, conditioned on $h_{j_n}$ being the most valuable allocated house, the minimum subsidy required to achieve envy-freeness is attained when $\{j_1,\dots,j_{n-1}\} = \{j_n-(n-1), j_n-(n-2),\dots, j_n-1\}$.
This means that in order to compute a minimum-subsidy envy-free outcome, it suffices to check every set of $n$ consecutive houses according to the order $h_1,\dots,h_m$, and find the smallest subsidy required among these sets.
The number of such sets is $m-n+1$, so the computation can be done in polynomial time.

In order to optimize this computation, we can first compute the prefix sums $\sum_{k=1}^i u(h_k)$ for all $i\in [m]$ in $O(m)$ time, so that we can compute the sum $\sum_{k=p}^q u(h_k)$ for any $p,q\in [m]$ in $O(1)$ time.
Then, for each set of $m-n+1$ consecutive houses, we can compute the required subsidy in $O(1)$ time.
It follows that the algorithm can be implemented in $O(m)$ time.
\end{proof}

\section{Conclusion}

In this paper, we have studied the problem of achieving envy-freeness in house allocation through the use of subsidy, with the goal of minimizing the subsidy involved.
A natural follow-up direction is to investigate the effects of strategic behavior in fair house allocation with subsidy, as has been done in the fair allocation of indivisible goods with subsidy \citep{goko2024subsidy}.
Interestingly, we observe that no deterministic mechanism that always returns a minimum-subsidy envy-free outcome can be strategyproof, even in the case of two agents and two houses.

\begin{example}
    Consider a house allocation problem instance with two houses and two agents whose utilities for the houses are given by $u_1(h_1) = 25$, $u_1(h_2) = 75$, $u_2(h_1) = 0$, and $u_2(h_2) = 100$.
    Let $\mathcal{M}$ be any deterministic mechanism that always returns a minimum-subsidy envy-free outcome.

    In this problem instance, there are only two possible allocations: $(h_1,h_2)$ and $(h_2,h_1)$.
    One can check (similarly to Example~\ref{ex:envy-freeable}) that the allocation $(h_2,h_1)$ is not envy-freeable.
    On the other hand, for the allocation $(h_1,h_2)$, the minimum subsidy required to achieve envy-freeness is $50$, given by the subsidy vector $(s_1,s_2) = (50,0)$.
    Hence, $\mathcal{M}$ must return the allocation $(h_1,h_2)$ along with the subsidy vector $(50,0)$.
    Note that agent~$1$ receives a utility of $25 + 50 = 75$ from this outcome.

    Next, suppose that agent~$1$ misreports her utilities as $u'_1(h_1) = 20$ and $u'_1(h_2) = 80$.
    Again, the allocation $(h_2,h_1)$ is not envy-freeable.
    For the allocation $(h_1,h_2)$, the minimum subsidy required to achieve envy-freeness is $60$, given by the subsidy vector $(s'_1,s'_2) = (60,0)$.
    Hence, $\mathcal{M}$ must return the allocation $(h_1,h_2)$ along with the subsidy vector $(60,0)$.
    This results in a true utility of $25 + 60 = 85$ for agent~$1$.

    Hence, agent~$1$ benefits from misreporting, which means that $\mathcal{M}$ cannot be strategyproof.
\end{example}

Other possible directions include exploring approximation algorithms (or hardness of approximation) for minimizing the subsidy, investigating the combination of envy-freeness and \emph{Pareto efficiency} (i.e., the property that there is no other outcome with the same subsidy that makes at least one agent better off and no agent worse off), as well as deriving worst-case bounds on the amount of subsidy required to achieve envy-freeness or other fairness notions.

\subsection*{Acknowledgments}

This research/project is supported by the National Research Foundation, Singapore under its AI Singapore Programme (AISG Award No: AISG-PhD/2021-08-013).
The work is also supported by the Singapore Ministry of Education under grant number MOE-T2EP20221-0001 and by an NUS Start-up Grant.
We thank the anonymous reviewer for valuable feedback.

\bibliographystyle{plainnat}
\bibliography{abb,bib}

\end{document}